\newif\ifFull
\Fulltrue

\documentclass[11pt]{article}
\usepackage[margin=1in]{geometry}
\usepackage{graphicx}
\usepackage{cite}
\usepackage{mathptmx} % replaces deprecated \usepackage{times}
\usepackage{amsmath}
\usepackage{amssymb}
\usepackage{paralist}
\usepackage{hyperref}
\usepackage{verbatim}
\usepackage{wrapfig}

\newtheorem {theorem} {Theorem}

\newtheorem {lemma}[theorem] {Lemma}

\newtheorem {corollary}[theorem] {Corollary}

\graphicspath{{fig/}}

\newenvironment {proof}{\noindent \textbf {Proof:}}{\mbox{}\hfill \ensuremath {\square}\medskip}
\newenvironment {sketch}{\noindent \textbf {Proof sketch:}}{\mbox{}\hfill \ensuremath {\square}\medskip}

\newcommand{\tw}{\ensuremath{\tau}}
\newcommand{\ext}{\mathrm{ext}}
\newcommand{\etalii}{et~al.}

\title{Pixel and Voxel Representations of Graphs%
  \thanks{This work was started at the 2014 Bertinoro Workshop on
    Graph Drawing. We thank the organizers for creating an
    inspiring atmosphere and Sue Whitesides for suggesting the
    problem.  A. Wolff acknowledges support by the ESF EuroGIGA project
    GraDR (DFG grant Wo~758/5-1).}}

\author
{
	Md.~Jawaherul~Alam\thanks{University of Arizona, Department of Computer Science, \texttt{mjalam@cs.arizona.edu}}
\and
	Thomas~Bl{\"a}sius\thanks{Karlsruhe Institute of Technology, Department of Computer Science, \texttt{thomas.blaesius@kit.edu}}
\and
	Ignaz~Rutter\thanks{Karlsruhe Institute of Technology, Department of Computer Science, \texttt{rutter@kit.edu}}
\and
	Torsten~Ueckerdt\thanks{Karlsruhe Institute of Technology, Department of Mathematics, \texttt{torsten.ueckerdt@kit.edu}}
\and
	Alexander~Wolff\thanks{Universit\"at W\"urzburg, Department of Computer Science, \texttt{alexander.wolff@uni-wuerzburg.de}}
}

\begin{document}

\maketitle

\begin{abstract}
  We study % a new type of 
  % [I don't mind "new" since in previous work the shape of a "blob" 
  % was fixed and the size of a representation was never considered.]
  contact representations for graphs, which we call \emph{pixel
    representations} in 2D and \emph{voxel representations} in 3D.
  Our representations are based on the unit square grid whose cells we call
  pixels in 2D and voxels in 3D.  Two pixels are adjacent if they
  share an edge, two voxels if they share a face.  We call a connected
  set of pixels or voxels a \emph{blob}.  Given a graph, we represent
  its vertices by disjoint blobs such that two blobs contain adjacent
  pixels or voxels if and only if the corresponding vertices are
  adjacent.  We are interested in the size of a representation, which
  is the number of pixels or voxels it consists of.

  We first show that finding minimum-size representations is
  NP-complete.  Then, we bound representation sizes needed for certain
  graph classes.  In 2D, we show that, for $k$-outerplanar graphs with
  $n$ vertices, $\Theta(kn)$ pixels are always sufficient and
  sometimes necessary.  In particular, outerplanar graphs can be
  represented with a linear number of pixels, whereas general planar
  graphs sometimes need a quadratic number.  In 3D, $\Theta(n^2)$
  voxels are always sufficient and sometimes necessary for any
  $n$-vertex graph.  We improve this bound to $\Theta(n\cdot \tw)$
  for graphs of treewidth~$\tw$ and to $O((g+1)^2n\log^2n)$ for graphs
  of genus $g$.  In particular, planar graphs admit representations
  with $O(n\log^2n)$ voxels.
\end{abstract}

\section{Introduction}
\label{sec:introduction}

In Tutte's landmark paper ``How to draw a graph'', he introduces
barycentric coordinates as a tool to draw triconnected planar graphs.
Given the positions of the vertices on the outer face (which must be
in convex position), the positions of the remaining vertices are
determined as the solutions of a set of equations.  While
the solutions can be approximated numerically, and symmetries tend to
be reflected nicely in the resulting drawings, the ratio between the
lengths of the longest edge and the shortest edge is exponential in
many cases.  This deficiency triggered research directed towards
drawing graphs on grids of small size in both 2D and 3D for different
graph drawing paradigms; Brandenburg \etalii~\cite{BEGKLM03} listed
this as an important open problem.
In \textit{straight-line grid drawings}, the vertices are at integer
grid points and the edges are drawn as straight-line segments. Both
Schnyder~\cite{s-epgg-SODA90} and de
Fraysseix~\etalii~\cite{fpp-hdpg-Combinatorica90}, 
 gave algorithms for drawing any $n$-vertex planar graph on a grid
 of size $O(n) \times O(n)$. There has also been research towards
 drawing subclasses of planar graphs on small-area grids.
%Garg and Rusu~\cite{GR07} gave an algorithm to draw $n$-vertex outerplanar
% graph on grids with area $O(dn^{1.48})$, where $d$ is the maximum degree
% of a vertex. This bound on grid size was later improved
% to $O(n^{1.48})$ area by Di Battista and Frati~\cite{df-sedog-Algorithmica09}
% and then to $O(dn\log n)$ area by Frati~\cite{Fra12}. 
 For example, any $n$-vertex outerplanar graph can
 be drawn in area $O(n^{1.48})$~\cite{df-sedog-Algorithmica09}.
 Similar research has also been done for
 other graph drawing problems, such as \textit{polyline drawings},
 where edges can have bends~\cite{Bie11}, \textit{orthogonal
   drawings}, where edges are polylines consisting of only
 axis-aligned segments~\cite{Bie11,CGKT02}, and for drawing graphs in
 3D~\cite{DMW13,PTT97,p-cr3do-JDA08} % {FLW03}.

A \textit{bar visibility representation}~\cite{TT86} draws a graph in a different way:
 the vertices are horizontal segments and the edges are realized by
 vertical line-of-sights between corresponding segments. 
% Rosenstiehl and Tarjan~\cite{RT86} and Tamassia and
% Tollis~\cite{TT86} independently gave linear-time algorithms to
% compute visibility representations for planar 
% graphs on a grid of size at most $(2n-5)\times (n-1)$, which was later
% improved to $(\lfloor4n/3\rfloor-2)\times (n-1)$ by
% Fan~\etalii~\cite{FLLY07}. 
Improving earlier results, Fan \etalii~\cite{FLLY07} showed
that any planar graph admits a visibility representation of size
$(\lfloor4n/3\rfloor-2)\times (n-1)$.
Generalized visibility representations for non-planar graphs have been
considered in 2D~\cite{eklmw-b1vgo-JGAA14,Bran14}, and in 3D~\cite{BEFH+98}.
 In all these and many subsequent papers, the size of a drawing is measured
 as the area or volume of the bounding box.

% \todo{@Sascha: Wasn't there also a paper of Biedl
%  at GD'09?  There is also work in 3D!  (Patrignani GD'05?, Wood and
%  Dujmovi\'c, Thiele)}

 Yet another approach to drawing graphs are the so-called
 \textit{contact representations}, where vertices are interior-disjoint geometric objects
 such as lines, curves, circles, polygons, polyhedra, etc. and edges
 correspond to pairs of objects touching in some specified way.
 An early work by Koebe~\cite{Koebe36} represents planar graphs with
 touching disks in 2D. Any planar graph can also be represented by contacts
 of triangles~\cite{FMR94}, by side-to-side contacts
 of hexagons~\cite{GHKK10}
% {BFM07} 
and of axis-aligned $T$-shape
 polygons~\cite{ourDCG13,FMR94}. 2D-contact representations of graphs
 with curves~\cite{Hli98}, line-segments~\cite{FM07-segment},
$L$-shapes~\cite{CKU13}, % {KUV13}, 
homothetic triangles~\cite{homothetic07},
 squares and rectangles~\cite{Felsner13,BGPV08} have also been studied.
Of particular interest are the so-called \textit{VCPG-representations} 
introduced by Aerts and Felsner~\cite{AF14}.  In such a representation,
vertices are represented by interior-disjoint paths in the plane
square grid and an edge is a contact between an endpoint of one path and
an interior point of another.  Aerts and Felsner showed that for
certain subclasses of planar graphs, the maximum number of bends per
path can be bounded by a small constant.

 Contact representations in 3D allow us to visualize non-planar graphs, but
 little is known about contact representations in 3D: Any planar
 graph can be represented by contacts of cubes~\cite{FF11}, and by face-to-face
 contact of boxes~\cite{BEF+12,Thom88}. Contact representations of
 complete graphs and complete bipartite graphs in 3D have been studied
 using spheres~\cite{BR13,HK01}, cylinders~\cite{Bez05}, and tetrahedra~\cite{Zong96}.
 In 3D as well as in 2D, the complexity of a contact representation is usually measured
 in terms of the \textit{polygonal complexity} (i.e., the number of
 corners) of the objects used in the representation.

In this paper, in contrast, we are interested in ``building'' graphs,
and so we aim at minimizing the cost of the building material---think
of unit-size Lego-like blocks that can be connected to each other
face-to-face. We represent each vertex by a connected set of building
blocks, which we call a \emph{blob}. If two
vertices are adjacent, the blob of one vertex contains a block
that is connected (face-to-face) to a block in the blob of the other.
 The blobs of two non-adjacent vertices are not connected.
% to each other face-to-face.
We call the building blocks \textit{pixels} in 2D and \textit{voxels} in 3D.
Accordingly, the 2D and 3D
 variants of such representations are called \textit{pixel} and \textit{voxel}
 representations, respectively. We define the \textit{size} of a pixel or
 voxel representation to be the total number of  boxes it consists of.
(We use \textit{box} to denote either pixel or voxel when the
dimension is not important.)
 
Although pixel representations can be seen as generalizations of
VCPG-repre\-sen\-ta\-tions where grid subgraphs instead of grid paths
are used, % to the best of our knowledge,
minimizing or bounding the size of such representations has not been
studied, so far, neither in 2D nor in 3D.

\paragraph{Our Contribution.}

We first investigate the complexity of our problem: finding
mini\-mum-size representations turns out to be NP-complete (%see
Section~\ref{sec:complexity}). Then, we give lower and upper bounds
for the sizes of 2D- and 3D-representations for certain graph classes:
\begin{itemize}
\item In 2D, we show that, for $k$-outerplanar graphs with $n$
  vertices, $\Theta(kn)$ pixels are always sufficient and sometimes
  necessary (see Section~\ref{sec:lower-upper-bounds-2d}).  In
  particular, outerplanar graphs can be represented with a linear
  number of pixels, whereas general planar graphs sometimes need a
  quadratic number.
\item In 3D, $\Theta(n^2)$ voxels are always sufficient and sometimes
  necessary for any $n$-vertex graph (see
  Section~\ref{sec:repr-3d-space}).  We improve this bound to
  $\Theta(n\cdot \tw)$ for graphs of treewidth~$\tw$ and to
  $O((g+1)^2n\log^2n)$ for graphs of genus $g$.  In particular, $n$-vertex planar
  graphs admit voxel representations with $O(n\log^2n)$ voxels.
\end{itemize}

\begin{comment}
\textbf{Question:} Can't we show that, in 2D (and similarly in 3D), we
can convert any straight-line drawing on a grid of size $w(n) \times
h(n)$ into a pixel representation using a grid of size $c \dot w(n)
\times c \dot h(n)$, for some universal constant $c>1$? \todo{Sascha: OK, I see that this
is probably not possible, even if we may assume that, among the
locations of the vertices, all $x$-coordinates ($y$-coordinates) are
pairwise disjoint.}\todo{MJA: I don't yet see why it is not possible, especially for orthogonal drawings.}
Then the size
of the pixel representation would be proportional to the total edge
length of the straight-line drawing.
\end{comment}

\bigskip

\section{Complexity}
\label{sec:complexity}

\ifFull

First, we show that it is NP-hard to compute 
minimum-size pixel representations. We reduce from the problem of
deciding whether a planar graph of maximum degree~4 has a
grid drawing where every edge has length~1.
Bhatt and Cosmadakis~\cite{bc-cmwlvl-87} showed that this problem is NP-hard
(even if the graph is a binary tree).
Their proof still works if the angles between adjacent edges are specified.
 Note that this also prescribes the circular order of edges around vertices up to reversal.
%
% Afterwards, we show that computing an ink-minimal contact
% representation with unit squares in the plane reduces to minimizing
% the ink in contact representations with unit boxes in 3D.
%

\begin{figure}
  \centering
  \includegraphics[page=1]{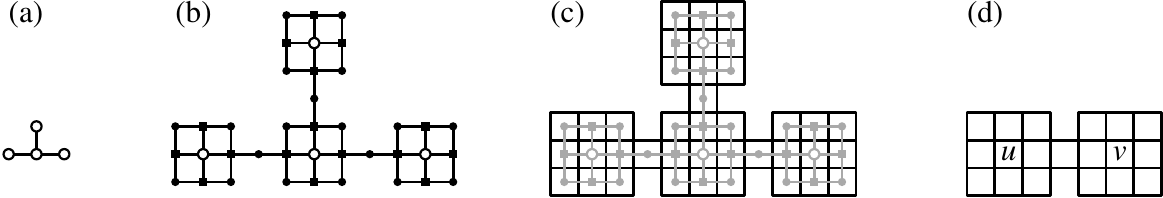}
  \caption{(a)~The graph $G$ with prescribed angles between edges.
    The edges in the drawing have length~1.  (b)~The graph $H$ drawn
    with edges of length~1.  (c)~Representation of $H$ with unit
    squares.  (d)~Unique representation with unit squares of a
    subgraph of $H$ corresponding to an edge $uv$ in $G$.}
  \label{fig:complexity}
\end{figure}

\begin{theorem}
  \label{thm:complexity}
  % Computing an ink-minimal contact representation
  % with unit squares (2D) is NP-complete.
  It is NP-complete to minimize the size of a pixel representation of
  a planar graph.
\end{theorem}
\begin{proof}
  Clearly the corresponding decision problem is in NP, thus it remains
  to show NP-hardness. Let $G$ be a planar graph of maximum
  degree~4 and assume that the angles between adjacent edges are
  prescribed.  We define a graph~$H$ as follows (see
  Figs.~\ref{fig:complexity}a and~~\ref{fig:complexity}b).  First,
  replace every vertex by a
  wheel with five vertices such that the angles between the edges are
  respected.  Second, subdivide every edge except those that are
  incident to the center of a wheel.  We claim that $G$ admits a grid
  drawing with edges of length~1 (respecting the prescribed angles)
  if and only if $H$ admits a representation where every
  vertex is represented by exactly one pixel.

  Assume $G$ admits a grid drawing with edges of length~1.  Scaling
  the drawing by a factor of~4 and suitably adding the new vertices
  and edges clearly yields a drawing of $H$ with edges of length~1,
  such that two vertices have distance~1 only if they are adjacent;
  see Fig.~\ref{fig:complexity}b.  For every vertex~$v$ of~$H$, we
  create a pixel $P_v$ with $v$ at its center
  (Fig.~\ref{fig:complexity}c).  Clearly, for two adjacent vertices
  $u$ and $v$ in $H$, the pixels $P_u$ and $P_v$ touch as the edge
  $uv$ has length~1 in the drawing of $H$.  Moreover, two pixel $P_u$
  and $P_v$ touch only if $u$ and $v$ have distance~1 and thus only if
  $u$ and $v$ are adjacent.  Hence, this set of pixels is a pixel
  representation of~$H$.

  Conversely, assume $H$ admits a representation such that every
  vertex $v$ is represented by a single pixel.  Obviously, the
  subdivided wheel of size~4 has a unique representation (up to
  symmetries) consisting of a square of $3\times 3$ pixels.  Consider
  two adjacent vertices $u$ and $v$ of $G$.  Then there is a $3\times
  3$ square for $u$ and one for $v$.  As $u$ and $v$ adjacent in $G$,
  there must be a pixel representing the subdivision vertex on the
  edge $uv$ in $H$ that touches both $3\times 3$ squares (of $u$ and
  $v$) as in Fig.~\ref{fig:complexity}d.  Thus, the straight line from
  the center of the square representing $u$ to the center of the
  square representing $v$ is either horizontal or vertical and has
  length~4.  Hence, we obtain a drawing of $G$ where every edge has
  length~4.  Scaling this drawing by a factor of ${1}/{4}$ yields a
  grid drawing of~$G$ with edges of length~1.
\end{proof}

\fi

Next, we reduce 
computing mini\-mum-size pixel representations to
computing mini\-mum-size voxel representations.
\ifFull

\begin{theorem}
  % Computing an ink-minimal contact representation
  % with unit boxes (3D) is NP-complete.
  It is NP-complete to minimize the size of a voxel representation of
  a graph.
\end{theorem}
\begin{proof}
  Again, the corresponding decision problem is clearly in NP.  To show
  NP-hard\-ness, we reduce from the 2D case.  To this end, we build a
  rigid structure called \emph{cage} that forces the graph in which we
  are actually interested to be drawn in a single plane.

  \begin{figure}[tb]
    \centering
    \includegraphics[page=1]{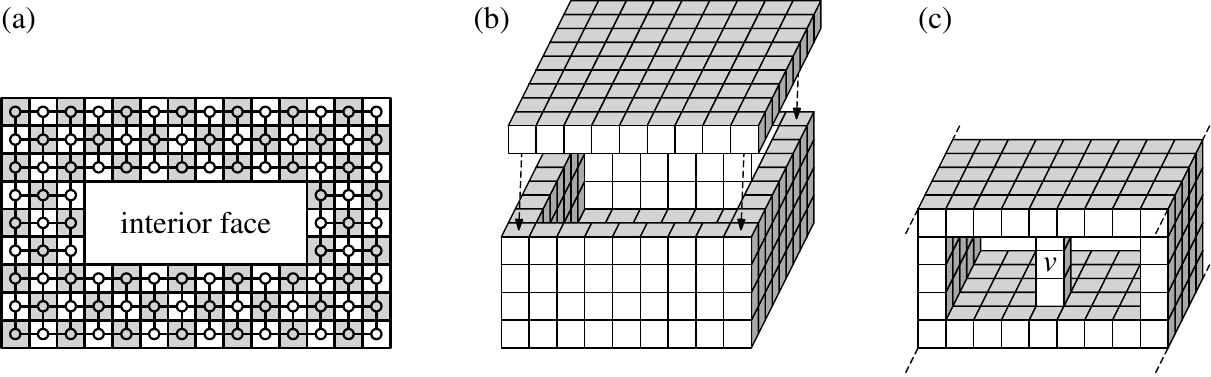}
    \caption{Illustration for the hardness proof in 3D. (a) A
      2-dimensional cage with thickness $3$ and interior face of size
      $8 \times 3$.  (b) A 3-dimensional cage with thickness $1$ and
      interior of size $7 \times 3 \times 7$. (c) Attaching $v$ to two
      sides of the box forces it into interior of cage.}
    \label{fig:complexity-3d}
  \end{figure}

  To simplify notation, we first prove for the 2-dimensional
  equivalent of a 3-dimensional cage that it actually is a rigid
  structure.  We then extend this to 3D.  The cage is basically the
  grid graph with a hole; see Fig.~\ref{fig:complexity-3d}a.  
  More precisely, the cage is defined
  by two parameters, the \emph{thickness} $t$, which is an integer,
  and by the \emph{interior} $w\times h$, which is a rectangle with
  integer width $w$ and integer height $h$.  Given these parameters,
  the corresponding cage is the graph obtained from the $(2t+w) \times
  (2t+h)$ grid by deleting a $w\times h$ grid such that the distance
  from the external face to the large internal face corresponding to
  the interior is~$t$.  We call this internal face the \emph{interior
    face}.  Fig.~\ref{fig:complexity-3d}a shows the cage with
  thickness~$3$ and interior $8\times 3$ together with a contact
  representation with exactly one pixel per vertex.

  Consider a pixel representation $\Gamma$ of the cage of
  thickness $t$ with interior $w\times h$.  We show that either the
  bounding box of the interior face has size at most $w\times h$ or
  $\Gamma$ uses at least one pixel per vertex plus $t$ additional
  pixels.  Thus, if we force some structure to lie in the
  interior of the cage, we can make the cost for using an area
  exceeding $w\times h$ arbitrarily large by increasing the thickness
  $t$ appropriately.

  We partition the cage into cycles $C_1, \dots, C_t$ where the
  vertices of $C_i$ have distance $i$ from the interior face.
  Consider $C_1$, which is the cycle bounding the interior face.  The
  cycle $C_1$ has four \emph{corner vertices} that are incident to two
  vertices in the outer face of~$C_1$.  All remaining vertices are
  incident to one vertex in the outer face.  Requiring $C_1$ to be
  represented with exactly one pixel per vertex such that the
  corner vertices have two sides and every other vertex has one side
  incident to the outer face implies that $C_1$ must form a rectangle
  of size $(w+1)\times(h+1)$.  Thus, if the bounding box of the
  interior face exceeds $w \times h$, $C_1$ requires at least one
  additional pixel.  Moreover, the bounding box of the outer
  face of $C_1$ exceeds $(w+1)\times(h+1)$.  Hence, an inductive
  argument shows that one requires at least one additional pixel
  for each of the cycles $C_1, \dots, C_t$, which shows the above
  claim.

  Analogously, we can build cages in 3D with thickness $t$ and
  interior $w \times h \times d$, by taking a 3D grid of size $(2t +
  w) \times (2t + h) \times (2t + d)$ and deleting a grid of size $w
  \times h \times d$.  Fig.~\ref{fig:complexity-3d}b shows the cage
  with~$t=1$, and $w\times h \times d = 7 \times 3 \times 7$.  Assume
  that we have a graph~$G$ for which we want to find an minimum-size pixel
  representation (in 2D).  We build a 3D cage, choose $t$, $w$ and
  $d$ to be very large, and set $h = 3$.  To force~$G$ to lie in the
  interior of the cage, we pick a vertex $v$ of $G$ and connect it to
  two vertices of the cage as shown in Fig.~\ref{fig:complexity-3d}c.
  This forces $G$ to completely lie in the interior of the cage.  As
  this interior has height~$3$ and no vertex of $G$ (except for $v$)
  is allowed to touch another vertex of the cage, $G$ is forced to lie
  in a single plane when choosing $t$ sufficiently large (obviously,
  polynomial size is sufficient). Moreover, choosing~$w$ and~$d$
  sufficiently large, ensures that the size of the plane available
  for~$G$ does not restrict the possible representations of~$G$.
%  If~$G'$ is the resulting graph,
  Finding a minimum-size pixel representation of~$G$ is
  equivalent to finding a minimum-size voxel representation of the resulting graph~$G'$.
\end{proof}

\else
In our reduction, we build a rigid structure around the given graph
that forces the given graph to be drawn in a single plane; see
Appendix~\ref{sec:app-NPC}. 

\wormholeThm{thm-voxel-NPC}
\begin{theorem}
\label{thm:voxel-NPC}
  It is NP-complete to minimize the size of a voxel representation of
  a graph.
\end{theorem}

\fi

\medskip

\section{Lower and Upper Bounds in 2D}
\label{sec:lower-upper-bounds-2d}

%A graph admits a pixel representation if and only if it is planar.
%Thus, we only consider planar graphs in this section.
Here we only consider planar graphs since only planar graphs
 admit pixel representations. Let $G$ be a planar graph with fixed plane
 embedding $\mathcal E$. The embedding $\mathcal E$ is \emph{$1$-outerplane}
 (or simply outerplane) if all vertices are on the outer face. It is \emph{$k$-outerplane}
 if removing all vertices on the outer face yields a $(k-1)$-outerplane embedding.
 A graph $G$ is \emph{$k$-outerplanar} if it admits a $k$-outerplane embedding but no
 $k'$-outerplane embedding for $k' < k$.  Note that $k \in O(n)$, where $n$
 is the number of vertices of $G$.

In Section~\ref{sec:lower-bound}, we show that pixel
representations of an $n$-vertex $k$-outerplanar graph sometimes
requires $\Omega(kn)$ pixels. As the number of pixels is a lower bound for
the area consumption, this strengthens a result by
 Dolev \etalii~\cite{dlt-pepg-84} that says that orthogonal drawings of planar
graphs of maximum degree~4 and \emph{width}~$w$ sometimes require
$\Omega(wn)$ area. As we will see later, width and $k$-outerplanarity
are very similar concepts.

In Section~\ref{sec:upper-bound}, we show that $O(kn)$ area and thus
using $O(kn)$ pixels is also sufficient. We use a result by
Dolev \etalii~\cite{dlt-pepg-84} who proved that any $n$-vertex planar
graph of maximum degree~4 and width~$w$ admits a planar orthogonal
drawing of area $O(wn)$.  The main difficulty is to extend their
result to general planar graphs.

\subsection{Lower Bound}
\label{sec:lower-bound}

Let $G$ be a $k$-outerplanar graph with a pixel representation~$\Gamma$. Note that
 a pixel representation $\Gamma$ induces an embedding of $G$. Let $\Gamma$ induce
a $k$-outerplane embedding of~$G$, which we call a
\emph{$k$-outerplane pixel representation} for short. 
We claim that the width and the height of~$\Gamma$ are at least $2k-1$.
 For $k = 1$ this is trivial as every (non-empty) graph requires width and height at least~$1$.
  For $k \geq 2$, let $V_\ext = \{v_1, \dots, v_\ell\}$ be the set of vertices incident to the outer face
 of $\Gamma$. Removing~$V_\ext$ from~$G$ yields a $(k-1)$-outerplane graph~$G'$
 with corresponding pixel representation~$\Gamma'$.  By induction, $\Gamma'$ requires
 width and height $2(k-1) - 1$.  As the representation of $V_\ext$ in $\Gamma$ encloses
 the whole representation~$\Gamma'$ in its interior, the width and the height of $\Gamma$
 are at least two units larger than the width and the height of $\Gamma'$, respectively.

%This shows the claim that $\Gamma$ has width and height at least $2k-1$.

Clearly, the number of pixels required by the vertices in~$V_\ext$ is at least the perimeter
 of~$\Gamma$ (twice the width plus twice the height minus~4 for the corners, which are shared)
 and thus at least $8k-8$. After removing the vertices in~$V_\ext$, the new vertices on the outer
 face require $8(k-1)-8$ pixels, and so on.  Thus, $\Gamma$ requires overall at least
 $\sum_{i=1}^{k}(8i-8) = 4k^2 -4k$ pixels, which gives the following lemma.

\begin{lemma}
  \label{lem:lower-bound-k-outerplanar}
  Any $k$-outerplane pixel representation has size at least
  $4k^2-4k$.
\end{lemma}

There are $k$-outerplanar graphs with $n$ vertices such that $k \in
\Theta(n)$.  For example, the nested triangle graph with $2k$ triangles 
(see Fig.~\ref{fig:nested-triangles}) has $n =
6k$ vertices and is $k$-outerplanar for $k \geq 2$.
% for $k=t/2$.
Let $G$ be a graph with $c$ connected components
% \todo{AW: Why do we need this parameter?
%   Aren't the nested triangles good enough?}\todo{TB: Without this, we
%   would only get $n$-outerplanar graphs that require $\Omega(n^2)$
%   pixels.  Now we for example also get $\sqrt{n}$-outerplanar graphs
%   with $\Omega(n^{1.5})$ pixels.}
  each of which is $k$-outerplanar
and has $\Theta(k)$ vertices.  Then each connected component requires
$4k^2-4k$ pixels (due to Lemma~\ref{lem:lower-bound-k-outerplanar}) and thus
we need at least $(4k^2-4k)c$ pixels in total.  As $G$ has $n =
\Theta(kc)$ vertices, we get $(4k^2-4k)c \in \Theta(kn)$, which proves the following.

\begin{theorem}
  Some $k$-outerplanar graphs require $\Omega(kn)$-size pixel representations.
\end{theorem}

\subsection{Upper Bound}
\label{sec:upper-bound}

In the following two lemmas, we first show how to construct a pixel
representation from a given orthogonal drawing and that taking minors
does not heavily increase the number of pixels we need.  Both lemmas
aim at extending a result of Dolev~\etalii~\cite{dlt-pepg-84} on
orthogonal drawings of planar graphs with maximum degree~4 to pixel
representations of general planar graphs.  As we re-use both lemmas in
the 3D case (Section~\ref{sec:repr-3d-space}), we state them in the
general $d$-dimensional setting.

\begin{lemma}
  \label{lem:drawing-to-box-representation}
  Let $G$ be a graph with $n$ vertices, $m$ edges, and an orthogonal
  drawing of total edge length~$\ell$ in $d$-dimensional space.  Then
  $G$ admits a $d$-dimensional representation of size $2\ell + n - m$.
\end{lemma}
\begin{proof}
  We first scale the given drawing $\Gamma$ of~$G$ by a factor of~$2$
  and subdivide the edges of $G$ such that every edge has length~1.
  Denote the resulting graph by~$G'$ and its drawing by~$\Gamma'$.  An
  edge~$e$ of length $\ell_e$ in~$\Gamma$ is represented by a path
  with $2\ell_e - 1$ internal vertices (the subdivision vertices).
  Thus, the total number of subdivision vertices is $2\ell - m$.  Due
  to the scaling, non-adjacent vertices in $G'$ have distance greater
  than~1 in~$\Gamma'$ (adjacent vertices have distance~1).  Thus,
  representing every vertex~$v$ by the
% \todo{in the introduction, def
%    box as generic term for pixel and voxel}
box having $v$ as center
  yields a representation of $G'$ with $2\ell + n - m$ boxes (one box
  per vertex of $G'$).  If we assign the boxes representing
  subdivision vertices to one of the endpoints of the corresponding
  edge, we get a representation of~$G$ with $2\ell + n - m$ boxes.
\end{proof}

\begin{figure}[tb]
\begin{minipage}[b]{.32\textwidth}
  \centering
  \includegraphics[page=1]{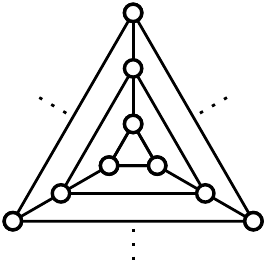}
  \caption{A nested triangle graph of outerplanarity $\Omega(n)$.}
  \label{fig:nested-triangles}
\end{minipage}
\hfill
\begin{minipage}[b]{.57\textwidth}
  \centering
  \begin{tabular}[b]{@{\hspace{-5.5ex}}c@{\hspace{2ex}}c@{\hspace{7ex}}c@{}}
    \includegraphics[page=1]{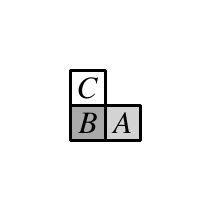} &
    \includegraphics[page=2]{minor} &
    \includegraphics[page=3]{minor} \\
    (a) & (b) & (c)
  \end{tabular}
  \caption{Constructing a representation of a minor with
    asymptotically the same number of blocks.}
  \label{fig:minor}
\end{minipage}
\end{figure}

\begin{lemma}
  \label{lem:minor}
  Let $G$ be a graph that has a $d$-dimensional representation of size
  $b$.  Every minor of~$G$ admits a $d$-dimensional representation of
  size at most $3^d b$.
\end{lemma}
\begin{proof}
  Let $H$ be a minor obtained from $G$ by first deleting some edges,
  then deleting isolated vertices, and finally contracting edges.  We
  start with the representation $\Gamma$ of $G$ using $b$ boxes and
  scale it by a factor of~3.  This yields a representation $3\Gamma$
  using $3^db$ boxes.  Then we modify~$3\Gamma$, without adding boxes,
  to represent the minor~$H$.  For convenience, we consider the 2D
  case; the case $d > 2$ works analogously.

  % \begin{figure}[tb]
  %   \centering
  %   \begin{floatrow}
  %     \ffigbox[0.33\textwidth]{
  %       \caption{A nested triangle graph with outerplanarity
  %         $\Omega(n)$.}
  %       \label{fig:nested-triangles}}{%
  %       \includegraphics[page=1]{fig/nested-ipe}
  %     }
  %     \ffigbox[0.62\textwidth]{
  %       \caption{Constructing a representation of a minor with
  %         asymptotically the same number of blocks.}
  %       \label{fig:minor}}{%
  %       \includegraphics[page=1]{minor}
  %     }
  %   \end{floatrow}
  % \end{figure}

  Let $uv$ be an edge in $G$ that is deleted.  In $3\Gamma$ we delete
  every pixel in the representation of $u$ that touches a pixel of the
  representation of $v$.  We claim that this neither destroys the
  contact of $u$ with any other vertex nor does it disconnect the
  shape representing $u$.  Consider a single pixel~$B$
  in~$\Gamma$.  In  
  $3\Gamma$ it is represented by a square of $3\times 3$ pixels
  \emph{belonging} to $B$.  If $B$ is in contact to another pixel $A$ in
  $\Gamma$, then there is a pair of pixels~$A'$ and~$B'$ in~$3\Gamma$
  such that $A'$ and $B'$ are in contact, while all other pixels that
  touch $A'$ and $B'$ belong to $A$ and $B$, respectively; see
  Figs.~\ref{fig:minor}a and~\ref{fig:minor}b.  Assume that we remove
  in~$3\Gamma$ all pixels belonging to $B$ that are in contact to pixels
  belonging to another pixel~$C$ touching~$B$ in~$\Gamma$; see
  Fig.~\ref{fig:minor}c.  Obviously, this does not effect the contact
  between~$A'$ and~$B'$.  Moreover, the remaining pixels belonging to
  $B$ form a connected blob.  The above claim follows immediately.

  Removing isolated vertices can be done by simply removing their
  representation.  Moreover, contracting an edge $uv$ into a vertex
  $w$ can be done by merging the blobs representing $u$ and $v$ into a
  single blob representing $w$.  This blob is obviously connected and
  touches the blob of another vertex if and only if either $u$ or $v$
  touch this vertex.
\end{proof}

Now let $G$ be a $k$-outerplanar graph.  Applying the algorithm of
Dolev~\etalii~\cite{dlt-pepg-84} yields an orthogonal drawing of total
length $O(wn)$, where~$w$ is the \emph{width} of~$G$.  The width~$w$
of~$G$ is the maximum number of vertices contained in a shortest 
path from an arbitrary vertex of~$G$ to a vertex on the outer face.
Given the orthogonal drawing, 
Lemma~\ref{lem:drawing-to-box-representation} gives us a pixel
representation of~$G$.  There are, however, two issues.
 First, $k$ and $w$ are not the same (e.g.,
subdividing edges increases $w$ but not $k$).  Second, $G$ does not
have maximum degree~4, thus we cannot simply apply the algorithm of
Dolev~\etalii~\cite{dlt-pepg-84}.

Concerning the first issue, we note that the algorithm of Dolev~\etalii\ exploits that~$G$ has width~$w$ only to find a
special type of separator~\cite[Theorem 1]{dlt-pepg-84}.
%For this, it is sufficient that $G$ is a subgraph of a graph of width~$w$.  It is
%not necessary that this supergraph has maximum degree~4 (in fact,
%Dolev~\etalii\ triangulate the given graph before finding the
%separator).
For this, it is sufficient that $G$ is a subgraph of a graph of width~$w$
 (not necessarily with maximum degree~4; in fact Dolev~\etalii\ triangulate
 the graph before finding the separator).

\begin{lemma}
  \label{lem:k-outerplanar-and-width}
  Every $k$-outerplanar graph has a planar supergraph of width $w =
  k$.
\end{lemma}
\begin{proof}
  Let $G$ be a graph with a $k$-outerplane embedding.  Iteratively
  deleting the vertices on the outer face gives us a sequence of
  deletion phases.  For each vertex~$v$, let $k_v$ be the phase in
  which $v$ is deleted.  Note that the maximum over all values
  of~$k_v$ is exactly~$k$.  For any vertex~$v$,
% it either holds that  $k_v = 1$ or that there is a vertex~$u$ with $k_u=k_v-1$
% such  that~$u$ and~$v$ are incident to a common face.
 either  $k_v = 1$ or there is a vertex~$u$ with $k_u=k_v-1$
 such that~$u$ and~$v$ are incident to a common face.
Thus, there is a
  sequence $v_1, \dots, v_{k_v}$ of $k_v$ vertices such that (i)~$v_1
  = v$, (ii)~$v_{k_v}$ lies on the outer face, and (iii)~$v_i$,
  $v_{i+1}$ are incident to a common face.  If the graph $G$ was
  triangulated, this would yield a path containing $k_v$ vertices from
  $v$ to a vertex on the outer face.  Thus, triangulated
  $k$-outerplanar graphs have width $w = k$.

  It remains to show that $G$ can be triangulated without
  increasing~$k_v$ for any vertex~$v$.  Consider a face $f$ and let
  $u$ be the vertex incident to $f$ for which $k_u$ is minimal.  Let
  $v\not=u$ be any other vertex incident to~$f$.  Adding the edge~$uv$
  clearly does not increase the value~$k_x$ for any vertex $x$.
  We add edges in this way until the graph is triangulated.
  Alternatively, we can use a result of Biedl~\cite{b-tkog-13} to
  triangulate~$G$.  Note that we do not need to triangulate the outer
  face of~$G$.  Hence, we do not increase the outerplanarity.
\end{proof}

To solve the second issue (the $k$-outerplanar graph $G$ not having
maximum degree~4), we construct a graph $G'$ such that $G$ is a minor
of $G'$, $G'$ is $k$-outerplanar, and $G'$ has maximum degree~4.  Then,
(due to Lemma~\ref{lem:k-outerplanar-and-width}) we can apply the
algorithm of Dolev~\etalii~\cite{dlt-pepg-84} to~$G'$.  Next, we
apply Lemma~\ref{lem:drawing-to-box-representation} to the resulting
drawing to get a representation of~$G'$ with $O(kn)$ pixels.
As~$G$ is a minor of~$G'$, Lemma~\ref{lem:minor} yields a 
representation of~$G$ that, too, requires $O(kn)$ pixels.

\begin{theorem}
  \label{thm:2d-upper-bound}
  Every $k$-outerplanar $n$-vertex graph has a size $O(kn)$ pixel representation.
\end{theorem}
\begin{proof}
  Let $G$ be a $k$-outerplanar graph.  After the above considerations,
  it remains to construct a $k$-outerplanar graph $G'$ with maximum
  degree~4 such that $G$ is a minor of $G'$.  Let $u$ be a vertex with
  $\deg(u) > 4$.  We replace $u$ with a path of length
  $\deg(u)$ and connect each neighbor of $u$ to a unique vertex of
  this path.  This can be done maintaining a plane
  embedding.  We now show that the resulting graph
  remains $k$-outerplanar.
  
  \begin{figure}[tb]
    \centering
    \includegraphics[page=1]{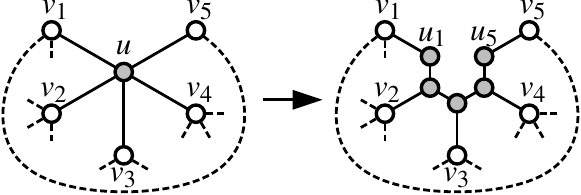}
    \caption{Replacement of high-degree vertices while preserving
      $k$-outerplanarity.}
    \label{fig:k-outerplanar}
  \end{figure}

  Consider a vertex $u$ on the outer face with neighbors $v_1, \dots,
  v_\ell$.  Assume the neighbors appear in that order around $u$ such
  that $v_1$ is the counter-clockwise successor of~$u$ on the outer
  face; see Fig.~\ref{fig:k-outerplanar}.  We replace $u$ with the
  path $u_1, \dots, u_\ell$ and connect $u_i$ to $v_i$ for $1 \le i
  \le \ell$.  Call the resulting graph $G_u$.  Note that all $u_i$ in
  $G_u$ are incident to the outer face.  Thus, if $G$ was
  $k$-outerplanar, $G_u$ is also $k$-outerplanar.  Moreover, the
  degrees of the new vertices do not exceed~4 (actually not even~3),
  and $G$ is a minor of $G_u$---one can simply contract the
  inserted path to obtain $G$.

  We can basically apply the same replacement if $u$ is not incident
  to the outer face.  Assume that we delete $u$ in phase $k_u$ if we
  iteratively delete vertices incident to the outer face.  When
  replacing $u$ with the vertices $u_1, \dots, u_\ell$, we have to
  make sure that all these vertices get deleted in phase $k_u$.  Let
  $f$ be a face incident to $u$ that is merged with the outer face
  after $k_u - 1$ deletion phases (such a face must exist, otherwise
  $u$ is not deleted in phase $k_u$).  We apply the same replacement
  as for the case where $u$ was incident to the outer face, but this
  time we ensure that the new vertices $u_i$ are incident to the
  face~$f$.  Thus, after $k_u - 1$ deletion phases they are all
  incident to the outer face and thus they are deleted in phase $k_u$.
  Hence, the resulting graph $G_u$ is $k$-outerplanar.  Again the new
  vertices have degree at most~3 and $G$ is obviously a minor of
  $G_u$.  Iteratively applying this kind of replacement for ever
  vertex $u$ with $\deg(u) > 4$ yields the claimed graph $G'$.

  The corresponding drawing can then be obtained as follows.  Since
  $G'$ has a supergraph of width $w=k$ by
  Lemma~\ref{lem:k-outerplanar-and-width}, and $G'$ has
  maximum degree~4, we use the algorithm of
  Dolev~\etalii~\cite{dlt-pepg-84} to obtain a drawing of~$G'$ with
  area (and hence total edge length) $O(nk)$.  By
  Lemma~\ref{lem:drawing-to-box-representation}, we thus obtain a
  representation of $G'$ with $O(nk)$ pixels.  Since
  $G$ is a minor of $G'$, Lemma~\ref{lem:minor} yields a
  representation of $G$ with $O(nk)$ pixels.
\end{proof}

\medskip

\section{Representations in 3D}
\label{sec:repr-3d-space}

In this section, we consider voxel representations.  We start with
some basic considerations showing that every $n$-vertex graph admits a
representation with $O(n^2)$ voxels.  Note that $\Omega(n^2)$ is
obviously necessary for~$K_n$ as every edge corresponds to a 
face-to-face contact and every voxel has at most $6$ such contacts.
We improve on this simple general result in two ways.
First, we show that $n$-vertex graphs with treewidth at most~$\tw$ admit 
voxel representations of size $O(n\cdot \tw)$
(see Section~\ref{sec:graphs-with-bounded-treew}).  Second, for $n$-vertex graphs
with genus at most $g$, we obtain representations with
$O(g^2n\log^2 n)$ voxels (see Section~\ref{sec:graphs-with-bounded-genus}).

% \subsection{General Graphs}
% \label{sec:general-graphs}

\begin{theorem}
  \label{thm:3d-general}
  Any $n$-vertex graph admits a voxel representation of size $O(n^2)$.
\end{theorem}
\begin{proof}
  Let $G$ be a graph with vertices $v_1, \dots, v_n$.
  Vertex~$v_i$ ($i=1,\dots,n$) is represented by three cuboids (see
  Fig.~\ref{fig:3-general}a), namely a vertical cuboid consisting of
  the voxels centered at the points $(2i, 2, 0), (2i, 3, 0), \dots,
  (2i, 2n, 0)$, a horizontal cuboid consisting of the voxels centered
  at $(2, 2i, 2), (3, 2i, 2), \dots, (2n, 2i, 2)$, and the voxel
  centered at $(2i, 2i, 1)$. This yields a representation where every
  vertex is a connected blob and no two
  blobs are in contact.  Moreover, for every pair of vertices~$v_i$
  and~$v_j$, there is a voxel of $v_i$ at $(2i, 2j, 0)$ and a voxel of
  $v_j$ at $(2i, 2j, 2)$ and no voxel between them at $(2i, 2j, 1)$.
  Thus, one can easily represent an arbitrary edge $(v_i, v_j)$ by
  extending the representation of $v_i$ to also contain $(2i, 2j, 1)$;
  see Fig.~\ref{fig:3-general}b.  Clearly, this representation
  consists of $O(n^2)$ voxels.
  % Let $G=(V,E)$ be a graph with vertex set $V = {1,\dots,n}$.  We
  % represent vertex $i$ by the union of the following straight-line
  % segments: the segment from $(i,0,0)$ to $(i,n,0)$, the segment from
  % $(0,i,1)$ to $(n,i,1)$ and the segment from $(i,i,0)$ to $(i,i,1)$.
  % For each edge ${i,j}$ with $i<j$ we add the segment $(i,j,0)$ to
  % $(i,j,1)$ to the representation of $i$, thus making the
  % representations of $i$ and $j$ touch; see Fig.~\ref{fig:3-general}.
\end{proof}

\begin{figure}[tb]
  \centering
  \includegraphics{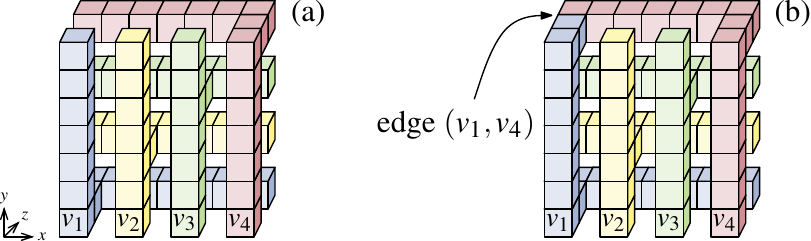}
  \caption{(a)~The basic contact representation without any contacts
    between vertices.  (b)~If~$v_1$ and~$v_4$ are adjacent,
    it suffices to add a single voxel to the representation of~$v_1$
    (or to that of~$v_4$).
    % Contact representation of general graphs with $O(n^2)$ ink.  The
    % vertical segments (blue) have $z$-coordinate $0$, the horizontal
    % segments (red) have $z$ coordinate $1$.  The boxes represent
    % axis-parallel segments of length~1 connecting the planes $z=0$ and
    % $z=1$.  An edge between the $i$th and the $j$th vertex with $i<j$
    % is represented by adding such a pillar at position $(i,j)$.
  }
  \label{fig:3-general}
\end{figure}

% \todo{do we need that?}Note that the representation spans only 2
% planes in \todo{vertical?}vertical direction.  We will make use of
% this to prove the next theorem.

\subsection{Graphs of Bounded Treewidth}
\label{sec:graphs-with-bounded-treew}

Let $G = (V, E)$ be a graph.  A \emph{tree decomposition} of $G$ is a
tree $T$ where each node~$\mu$ in $T$ is associated with a \emph{bag} $X_\mu
\subseteq V$ such that: (i) for each $v \in V$, the nodes of $T$ whose
 bags contain $v$ form a connected subtree, and (ii) for each edge $uv \in E$,
 $T$ contains a node $\mu$ such that $u, v\in X_\mu$.

\begin{comment}
 the following holds.
%
\begin{itemize}
\item For each $v \in V$, the nodes of $T$ whose bags contain $v$
  form a connected subtree.
\item For each edge $uv \in E$, $T$ contains a node $\mu$ such that
  $u, v\in X_\mu$.
\end{itemize}
\end{comment}
%
Note that we use (lower case) Greek letters for the nodes of~$T$ to
distinguish them from the vertices of $G$.  The \emph{width} of the
tree decomposition is the maximum bag size minus~$1$.  The
\emph{treewidth} of~$G$ is the minimum width over all
tree decompositions of~$G$.
A tree decomposition is \emph{nice} if $T$ is a rooted binary tree,
 where for every node $\mu$:
%such that for every node $\mu$ of $T$ it holds that
%
\begin{itemize}
\item $\mu$ is a leaf and $|X_\mu| = 1$ (\emph{leaf node)}, or
\item $\mu$ has a single child $\eta$ with $X_\mu \subseteq X_\eta$
  and $|X_\mu| = |X_\eta| - 1$ (\emph{forget node}), or
\item $\mu$ has a single child $\eta$ with $X_\eta \subseteq X_\mu$
  and $|X_\mu| = |X_\eta| + 1$ (\emph{introduce node}), or
\item $\mu$ has two children $\eta$ and $\kappa$ with $X_\mu = X_\eta
  = X_\kappa$ (\emph{join node}).
\end{itemize}
Any tree decomposition can be transformed (without increasing its
width) into a nice tree decomposition such that the resulting tree $T$
has $O(n)$ nodes, where $n$ is the number of vertices of
$G$~\cite{b-t-97}.  This transformation can be done in linear time.
% Any tree decomposition~$T$ can be easily transformed (without
% increasing the width) into a tree decomposition such that $T$ has
% $O(n)$ nodes.  \todo{Is this also shown in~\cite{b-t-97}?}\todo{no}
% Moreover, from a given tree decomposition, one can construct a nice
% tree decomposition (in linear time with \todo{given? resulting in?}  a
% tree of linear size) of the same width~\cite{b-t-97}. 
Thus, we can assume any tree decomposition to be a nice tree
decomposition with a tree of size $O(n)$.
% For nice tree decompositions, we have the following lemma.

\begin{lemma}
  \label{lem:tree-decomp-star}
  Let $T$ be a nice tree decomposition of a graph $G$.  The edges of
  $G$ can be mapped to the nodes of $T$ such that every edge $uv$ of
  $G$ is mapped to a node $\mu$ with $u, v \in X_\mu$ and the edges
  mapped to each node $\mu$ form a star.
\end{lemma}
\begin{proof}
  We say that a node $\mu$ \emph{represents} the edge $uv$ if $uv$ is
  mapped to $\mu$.  Consider a node $\mu$ during a bottom-up traversal
  of $T$.  We want to maintain the invariant that, after
  processing $\mu$, all edges between vertices in $X_\mu$ are
  represented by~$\mu$ or by a descendant of~$\mu$.  This ensures that
  every edge is represented by at least one node.  Every edge can then
  be mapped to one of the nodes representing it.

  If $\mu$ is a leaf, it cannot represent an edge as $|X_\mu| = 1$.
  If $\mu$ is a forget node, it has a child $\eta$ with $X_\mu
  \subseteq X_\eta$.  Thus, by induction, all edges between vertices
  in $X_\mu$ are already represented by descendants of $\mu$.  If
  $\mu$ is an introduce node, it has a child $\eta$ and $X_\mu =
  X_\eta \cup \{u\}$ for a vertex $u$ of $G$.  By induction, all edges
  between nodes in $X_\eta$ are already represented by descendants of
  $\mu$.  Thus, $\mu$ only needs to represent the edges between the
  new node $u$ and other nodes in $X_\mu$.  Note that these edges form
  a star with center $u$.  Finally, if $\mu$ is a join node, no edge
  needs to be represented by $\mu$ (by the same argument as for forget
  nodes).  This concludes the proof.
\end{proof}

We obtain a small voxel representation roughly as follows.  We start
with a ``2D'' voxel representation of the tree $T$, that is, all voxel
centers lie in the $x$--$y$ plane.  We take $\tw + 1$ copies of this
representation and place them in different layers in 3D space.  We
then assign to each vertex $v$ of $G$ a piece of this layered
representation such that its piece contains all nodes of $T$ that
include $v$ in their bags.  For an edge $uv$, let $\mu$ be the node to
which $uv$ is mapped by Lemma~\ref{lem:tree-decomp-star}.  By
construction, the representation of $\mu$ occurs multiple times
representing $u$ and $v$ in different layers.  To represent $uv$, we
only have to connect the representations of~$u$ and~$v$.  As it
suffices to represent a star for each node $\mu$ in this way, the
number of voxels additionally used for these connections is small.

\begin{theorem}
  \label{th:3d-treewidth}
  Any $n$-vertex graph of treewidth~\tw\ has a voxel representation
  of size $O(n\tw)$.
\end{theorem}
\begin{proof}
  Let $G$ be an $n$-vertex graph of treewidth~\tw.
  During our construction, we will get some contacts between the blobs
  of vertices that are actually not adjacent in~$G$.  
  As $G$ is a minor of the graph that we
  represent this way, we can use Lemma~\ref{lem:minor} to get
  a representation of $G$.  Let~$T$ be a nice tree decomposition
  of~$G$.  As a tree, $T$ is outerplanar and, hence, admits
  a pixel representation~$\Gamma$ with $O(n)$ pixels (by
  Theorem~\ref{thm:2d-upper-bound}).  Let~$\Gamma_1,\dots,\Gamma_k$ be
  voxel representations corresponding to~$\Gamma$ with
  $z$-coordinates $1,\dots, k = \tw + 1$.
 
  For a vertex $v$ of $G$, we denote by $\Gamma_i(v)$ the
  sub-representation of $\Gamma_i$ induced by the nodes of $T$ whose
  bags contain $v$.  Now let $c \colon V \to \{1, \dots, k\}$ be a
  $k$-coloring of $G$ with color set $\{1,\dots,k\}$ such
  that no two vertices sharing a bag have the same color.  Such a
  coloring can be computed by traversing $T$ bottom up, assigning in
  every introduce node $\mu$ a color to the new vertex that is not
  already used by any other vertex in~$X_\mu$.  As a basis for our
  construction, we represent each vertex $v$ of $G$ by the
  sub-representation~$\Gamma_{c(v)}(v)$.

  So far, we did not represent any edge of~$G$. Our
  construction, however, has the following properties: (i)~it uses
  $O(nk)$ voxels.  (ii)~every vertex is a connected set of
  voxels.  (iii)~for every node~$\mu$ of~$T$, there is a position 
  $(x_\mu, y_\mu)$ in the plane such that, for every vertex $v \in X_\mu$, 
  the voxel at $(x_\mu, y_\mu, c(v))$ belongs to the representation of~$v$.
  Scaling the representation by a factor of~$2$ ensures that
  this is not the only voxel for~$v$ and that $v$ is not
  disconnected if this voxel is removed (or reassigned to
  another vertex).

  By Lemma~\ref{lem:tree-decomp-star} it suffices to represent for
  every node $\mu$ edges between vertices in $X_\mu$ that form a star.
  Let $u$ be the center of this star.  We simply assign the voxels
  centered at $(x_\mu, y_\mu, 1), \dots, (x_\mu, y_\mu, k)$ to the
  blob of~$u$.  This creates a contact between $u$ and every other
  vertex $v \in X_\mu$ (by the above property that the voxel $(x_\mu,
  y_\mu, c(v))$ belonged to~$v$ before).  Finally, we apply
  Lemma~\ref{lem:minor} to get rid of unwanted contacts.  The
  resulting representation uses $O(nk)$ voxels, which concludes the
  proof.
\end{proof}

Note that cliques of size $k$ require $\Omega(k^2)$ voxels.  Taking the
disjoint union of $n/k$ such cliques yields graphs with $n$ vertices
requiring $\Omega(nk)$ voxels.  Note that these graphs have treewidth $\tw
= k - 1$.  Thus, the bound
of Theorem~\ref{th:3d-treewidth} is asymptotically tight.

\begin{theorem}
  Some $n$-vertex graphs of treewidth~$\tw$ require %size-
  $\Omega(n\tw)$ voxels. % representations.
\end{theorem}

\subsection{Graphs of Bounded Genus}
\label{sec:graphs-with-bounded-genus}

\begin{comment}
As planar graphs (that is, graphs of genus~$0$) have treewidth
$O(\sqrt{n})$~\cite{ft-nubdpg-06}, we obtain the following corollary
from Theorem~\ref{th:3d-treewidth}.

\begin{corollary} 
  Any planar graph $G$ with $n$ vertices admits a voxel representation
  of size $O(n^{1.5})$.
\end{corollary}
\end{comment}
Since planar graphs (genus~$0$) have treewidth $O(\sqrt{n})$~\cite{ft-nubdpg-06},
we can obtain a voxel representation of size $O(n^{1.5})$ for any planar graph,
% as a corollary
 from Theorem~\ref{th:3d-treewidth}.
Next, we improve this bound to $O(n \log^2 n)$ by proving
a more general result for graphs of bounded genus.
Recall that we used known results on orthogonal drawings with small
area to obtain small pixel representations in
Section~\ref{sec:upper-bound}. Here we follow a similar
approach (re-using
Lemmas~\ref{lem:drawing-to-box-representation} and~\ref{lem:minor}), 
now allowing the orthogonal drawing we 
start with to be non-planar.

We obtain small voxel representations by first showing that it is
sufficient to consider graphs of maximum degree~4: we replace
higher-degree vertices by connected subgraphs as in the proof of
Theorem~\ref{thm:2d-upper-bound}. Then we use a result of
Leiserson~\cite{l-aeglv-80} who showed that any graph of genus~$g$
and maximum degree~4 admits a 2D orthogonal drawing of area $O((g+1)^2
n \log^2n)$, possibly with edge crossings.  The
area of an orthogonal drawing is clearly an upper bound for its total
edge length.  Finally we turn the pixels into voxels and use the third
dimension to get rid of the crossings without using 
too many additional voxels. 

\begin{theorem}
  \label{thm:bounded-genus}
  Every $n$-vertex graph of genus $g$ admits a voxel representation of
  size $O((g+1)^2 n \log^2n)$.
\end{theorem}
\begin{proof}
  Let $G$ be an $n$-vertex graph, and let $u$ be a vertex of degree
  $\ell > 4$.  Assume~$G$ to be embedded on a surface of
  genus $g$, and let $v_1, \dots, v_\ell$ be the neighbors of $u$
  appearing in that order around $u$ (with respect to the embedding).
  We replace $u$ with the cycle $u_1, \dots, u_\ell$ and connect~$u_i$
  to~$v_i$ for $1 \le i \le \ell$; see Fig.~\ref{fig:3d-max-deg-4}a.
  Clearly, the new vertices have degree~3 and the genus of the graph
  has not increased.  Applying this modification to every vertex of
  degree at least~5 yields a graph $G_4$ of maximum degree~4 and genus
  $g$.  Moreover, $G$ is a minor of $G_4$ as one can undo the
  cycle replacements by contracting all edges in the cycles.  Thus, we
  can transform a voxel representation of~$G_4$ into a voxel
  representation of~$G$ by applying Lemma~\ref{lem:minor}.

  We claim that the number $n_4$ of vertices in $G_4$ is linear in $n$.
  Indeed, if $m$ denotes the number of edges in $G$, then we have $n_4 \leq n + 2m$.
%   Our claim holds for the following reason.
%   The number~$n_4$ of vertices of~$G_4$ is linear in the number~$n$ of
%   vertices of~$G$ plus the number~$m$ of edges of~$G$.
  Moreover, we can assume without loss of generality that $g \in O(n)$ (otherwise
  Theorem~\ref{thm:3d-general} already gives a better bound).  This
  implies that $m \in O(n)$ and hence, $n_4 \in O(n)$, as we claimed.

  We thus assume that $G$ has maximum degree~4. Then $G$ has a 
  (possibly non-planar) orthogonal drawing $\Gamma$
  of total edge length $O(g^2 n \log^2n)$~\cite{l-aeglv-80}.
%  Leiserson~\cite{l-aeglv-80} showed that any maximum-degree-4 graph
%  has a (possibly non-planar) orthogonal drawing \todo{Running time?}
%  of total edge length $O(g^2 n \log^2n)$.  Let~$\Gamma$ be the
%  resulting drawing of~$G$. 
  We modify~$G$ and~$\Gamma$ as follows. For every bend on an edge $e$ in $\Gamma$, we
  subdivide the edge $e$ once yielding a partition of the edges 
  of the subdivided graph into horizontal and vertical
  edges. We obtain a graph $G'$ from this subdivision of $G$ by replacing 
  every vertex $v$ by two adjacent vertices $v_1$ and $v_2$, and connecting
  $v_1$ and $w_1$ (respectively $v_2$ and $w_2$) by an edge if $v$ and $w$ are
  connected by a horizontal (respectively vertical edge); see Fig.~\ref{fig:3d-max-deg-4}b.
    
%   We denote the graphs induced by these two edge sets by
%   $G_1$ and $G_2$, respectively; 
%   \todo{What about calling the two graphs $G_0$ and $G_1$ instead of
%     $G^{-}$ and $G^{|}$?} 

\begin{figure}[tb]
  \centering
  \includegraphics[page=1]{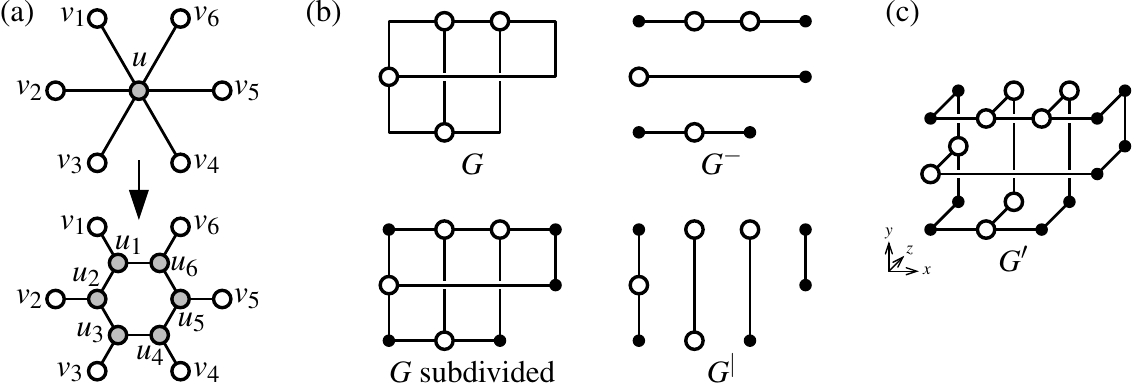}
  \caption{Constructing voxel representations for bounded-genus
    graphs: (a)~replacing high-degree vertices while preserving the
    genus, (b)~subdividing and decomposing a graph according to a
    non-planar orthogonal drawing with small area, and
    (c)~constructing a 3D drawing with small total edge length from
    the decomposition in (b).}
  \label{fig:3d-max-deg-4}
\end{figure}
 
  We draw $G'$ in 3D space by using the drawing $\Gamma$ and setting for 
  every vertex $v$ the $z$-coordinate of $v_1$ and $v_2$ to $0$ and $1$, 
  respectively. The $x$- and $y$-coordinates of vertices and edges
  are the same as in $\Gamma$; see Fig.~\ref{fig:3d-max-deg-4}b.
 \begin{comment}
  We draw $G_1$ and $G_2$ together in 3D space by using the
  drawing $\Gamma$ and setting the $z$-coordinate of vertices in
  $G^{-}$ and $G^{|}$ to~0 and~1, respectively.  Note that any vertex
  of~$G$ that appears in both~$G^{-}$ and~$G^{|}$ is drawn twice.  The
  two copies have the same $x$- and $y$-coordinates.  Thus, we can
  connect the two copies by an 
  edge of length~1.  Denote the resulting graph by $G'$
  \todo{Introduce $G'$ earlier!} and its
  drawing by $\Gamma'$; see Fig.~\ref{fig:3d-max-deg-4}b.
 \end{comment}
  Note that $G$ is a minor of~$G'$: we obtain~$G$ from~$G'$ by contracting
  (i)~the edge $v_0v_1$ for every vertex $v$ and (ii)~any subdivision vertex.
  Asymptotically, the total edge
  length of~$\Gamma'$ is the same as that of~$\Gamma$, that is, 
  $O((g+1)^2 n \log^2n)$.
  By Lemma~\ref{lem:drawing-to-box-representation}, we turn~$\Gamma'$ into
  a voxel representation of~$G'$ and, by Lemma~\ref{lem:minor}, into a
  voxel representation of~$G$ with size $O((g+1)^2 n \log^2n)$.
%. The size remains $O(g^2 n \log^2n)$.
\end{proof}

\medskip

\section{Conclusion}

In this paper, we have studied pixel representations and voxel
representations of graphs, where vertices are represented by disjoint
blobs (that is, connected sets of grid cells) and edges correspond to
pairs of blobs with face-to-face contact.  We have shown that it is
NP-complete to minimize the number of pixels or voxels in such
representations.  Does this problem admit an approximation algorithm?

% In 2D, we show that for $n$-vertex $k$-outerplanar graphs
% $\Theta(nk)$ pixels are always sufficient and sometimes
% necessary. In 3D, we need $\Theta(n\tw)$ voxels for $n$-vertex
% graphs with treewidth $\tw$ and $O((g+1)^2n \log^2n)$ voxels for
% $n$-vertex graphs with genus $g$.

We have shown that $O((g+1)^2n \log^2n)$ voxels suffice for any
$n$-vertex graph of genus~$g$.  It remains open to improve this upper
bound or to give a non-trivial lower bound.  % In particular,
We believe that any planar graph admits a voxel representation of
linear size.

\bibliographystyle{abbrv}
\bibliography{contact}

\end{document}